\documentclass[conference]{IEEEtran}
\usepackage[utf8]{inputenc}
\IEEEoverridecommandlockouts
\usepackage{cite}
\usepackage{amsmath,amssymb,amsfonts}
\usepackage{algorithmic}
\usepackage{graphicx}
\usepackage{textcomp}
\usepackage{xcolor}
\usepackage{subcaption}
\usepackage{amsmath,amssymb,amsthm}

\theoremstyle{definition}
\newtheorem{definition}{Definition}

\theoremstyle{plain}
\newtheorem{proposition}{Proposition}
\newtheorem{corollary}{Corollary}
\newtheorem{theorem}{Theorem}  
\usepackage{tabularx,booktabs,array}
\usepackage{amsmath,amssymb,amsthm}
\usepackage{ragged2e}
\DeclareUnicodeCharacter{2212}{-}

\def\BibTeX{{\rm B\kern-.05em{\sc i\kern-.025em b}\kern-.08em
    T\kern-.1667em\lower.7ex\hbox{E}\kern-.125emX}}

\begin{document}

\title{Privacy-Preserving AI-Enabled Decentralized
Learning and Employment Records System\\}

\author{\IEEEauthorblockN{1\textsuperscript{st} Yuqiao Xu}
\IEEEauthorblockA{\textit{Dept. Computer and Data Sciences} \\
\textit{Case Western Reserve University}\\
Cleveland, United States \\
0009-0009-3552-2136}
\and
\IEEEauthorblockN{2\textsuperscript{nd} Mina Namazi}
\IEEEauthorblockA{\textit{Dept. Computer and Data Sciences} \\
\textit{Case Western Reserve University}\\
Cleveland, United States \\
0000-0002-8878-9362}
\and
\IEEEauthorblockN{3\textsuperscript{rd} Sahith Reddy Jalapally}
\IEEEauthorblockA{\textit{Dept. Computer and Data Sciences} \\
\textit{Case Western Reserve University}\\
Cleveland, United States \\
srj58@case.edu}
\and
\IEEEauthorblockN{4\textsuperscript{th} Osama Zafar}
\IEEEauthorblockA{\textit{Dept. Computer and Data Sciences} \\
\textit{Case Western Reserve University}\\
Cleveland, United States \\
0009-0008-9621-6899}
\and
\IEEEauthorblockN{5\textsuperscript{th} Youngjin Yoo}
\IEEEauthorblockA{\textit{Dept. Management} \\
\textit{The London School of Economics and Political Science}\\ London, United Kingdom \\
0000-0001-8548-3475}
\and
\IEEEauthorblockN{6\textsuperscript{th} Erman Ayday}
\IEEEauthorblockA{\textit{Dept. Computer and Data Sciences} \\
\textit{Case Western Reserve University}\\
Cleveland, United States \\
0000-0003-3383-1081}
}

\maketitle

\begin{abstract}
Learning and Employment Record (LER) systems are emerging as critical infrastructure for securely compiling and sharing educational and work achievements. Existing blockchain-based platforms leverage verifiable credentials but typically lack automated skill-credential generation and the ability to incorporate unstructured evidence of learning. In this paper, a privacy-preserving, AI-enabled decentralized LER system is proposed to address these gaps. Digitally signed transcripts from educational institutions are accepted, and verifiable self-issued skill credentials are derived inside a trusted execution environment (TEE) by a natural language processing pipeline that analyzes formal records (e.g., transcripts, syllabi) and informal artifacts. All verification and job–skill matching are performed inside the enclave with selective disclosure, so raw credentials and private keys remain enclave-confined. Job matching relies solely on attested skill vectors and is invariant to non-skill resume fields, thereby reducing opportunities for screening bias.

The NLP component was evaluated on sample learner data; the mapping follows the validated Syllabus-to-O*NET methodology, and a stability test across repeated runs observed $<\!5\%$ variance in top-ranked skills. Formal security statements and proof sketches are provided showing that derived credentials are unforgeable and that sensitive information remains confidential. The proposed system thus supports secure education and employment credentialing, robust transcript verification, and automated, privacy-preserving skill extraction within a decentralized framework.
\end{abstract}


\section{Introduction}
\label{intro}

The rapid digital transformation of education and employment has increased the demand for accurate, secure, and privacy-preserving credential verification. Digital credential platforms alone issued over 25 million verifiable credentials in 2023, reflecting a 26\% annual increase driven by the surge in online certifications and badges~\cite{accredible2023future}. Ensuring trustworthiness in credentials is thus essential: employers must reliably validate applicants' claimed skills, and individuals need portable credentials they can securely share.

However, traditional credential systems fall short of these requirements. Degrees and transcripts typically provide only coarse-grained information (e.g., course titles or a diploma) and do not capture the specific skills or competencies that learners actually acquired by learners. Informal achievements such as completions of massive open online courses(MOOCs), coding bootcamp certificates, or other self-issued micro-credentials are seldom acknowledged, leaving valuable skills undocumented and difficult to verify. Legacy verification processes require disclosure of entire transcripts or certificates, unnecessarily exposing unrelated personal information~\cite{opencreds_framework}. Such inefficiencies leave systems vulnerable to credential fraud. Indeed, studies indicate that over half of job applicants misrepresent their qualifications, and only about half of employers rigorously verify credentials~\cite{ams2024fraud,shrm2023skills}.

These challenges are compounded by the shift toward skill-based hiring, reportedly adopted by approximately 75\% of employers as of 2023~\cite{shrm2023skills}. Extracting and standardizing skill credentials remains difficult due to curriculum variability across institutions, which complicates direct comparison. For instance, a"Software Engineer" course at two universities might impart different content or depth of knowledge, making it hard to equate learning outcomes at face value. Initiatives like OpenCreds aim to present competencies more precisely ~\cite{opencreds_framework}, yet integration into secure, privacy-preserving frameworks remains limited.

Centralized credential platforms also introduce privacy and resilience risks by creating single points of failure, with high-profile breaches affecting larger user populations \cite{peetz2025powerschool}. Moreover, these platforms typically lack selective disclosure mechanisms, violating data-minimization principles by forcing users to disclose entire credential sets rather than only the attributes required for a given verification.

To address these challenges, we propose a decentralized LER system leveraging secure enclaves(TEEs) and a natural language processing(NLP) pipeline to enable privacy-preserving credential verification and skill matching. The user retains complete sovereignty over their credentials, including formal (degrees, certificates) and informal (microcredentials, self-attested skills) records, which are from online courses and BootCamp. Our architecture employs secure enclaves, trusted execution environments performing isolated computations, process credentials securely without exposing raw data externally. NLP methods within enclaves extract structured skill data from unstructured learning records. The system securely matches extracted credentials against job requirements, ensuring no centralized platform accesses complete learner histories or sensitive employer criteria. Crucially, our architecture supports selective disclosure, empowering users to share only relevant credentials securely.

In summary, this paper makes the following key contributions:
\begin{itemize}
 \item A lightweight secure enclave infrastructure enabling decentralized credential storage and secure computation for credential verification and matching.
 \item  An NLP pipeline, executed inside the enclave, that converts formal and supported informal learning records into verifiable, selectively disclosable skill credentials.
 \item A verifier-side matching mechanism that consumes only attested skill vectors; non-skill invariance is formalized as a design-level fairness property reducing opportunities for screening bias.
 \item A decentralized, user-centric architecture providing users full control over credential management and selective disclosure capabilities.
 \item A systems-focused evaluation reporting runtime characteristics and enclave resource profiles on sample learner data, with mapping aligned to validated Syllabus-to-O*NET methodology and a repeatability check; formal security statements and proof sketches addressing confidentiality and unforgeability.
 \end{itemize}

This paper is organized as follows. Section~\ref{related} reviews related work. Section~\ref{build} presents foundational building blocks. Section~\ref{setting} states the system model and threat assumptions. Section~\ref{proposed} details the framework and protocols. Section~\ref{analysis} provides security and privacy analysis. Section~\ref{evaluation} presents the evaluation (focused on a representative computer-science corpus; broader domains are deferred). Section~\ref{discussion} discusses implications and limitations, and Section~\ref{conc} concludes.

\section{Related Work}
\label{related}

Blockchain technology has emerged as a promising solution for LER systems, offering enhanced security, transparency, and efficiency~\cite{permana2024potential}. However, its adoption faces significant challenges such as scalability limitations, slow transaction speeds, and the Blockchain scaling trilemma, the difficulty of simultaneously achieving decentralization, security, and scalability~\cite{steiu2020blockchain,permana2024potential}. Educational institutions’ traditionally cautious adoption approaches further complicate implementation, driven by perceived risks and the relative immaturity of blockchain solutions~\cite{mohammad2022barriers}. Regulatory compliance with data protection laws (e.g., GDPR and the California Consumer Protection Act), financial constraints from high infrastructure and maintenance costs, and inherent privacy concerns due to blockchain’s immutable data storage collectively hinder widespread deployment~\cite{elkhodr2024systematic,haugsbakken2019blockchain,wang2025transforming,bhaskar2020blockchain}. Security and privacy concerns, especially around blockchain’s immutable data storage, compound these challenges, requiring institutions to carefully balance transparency and data protection \cite{wang2025transforming,bhaskar2020blockchain}.

Current blockchain-based credentialing frameworks, such as OpenCreds~\cite{opencreds_framework}
OpenWallet Foundation~\cite{open_wallet}, and the OpenWallet Foundation~\cite{open_wallet}, predominantly support issuer-based credentials and formal learning scenarios. They lack robust mechanisms for self-issued credentials or informal learning validation, limiting their applicability to diverse learning experiences beyond formal educational institutions~\cite{legicred_digital_transformation,dcc_framework}.

To address these gaps, our framework uniquely integrates decentralized verification via secure enclaves and advanced NLP-driven analytics. Unlike existing approaches, our system effectively handles both institutional and self-issued credentials, automatically extracting structured skills from formal and informal learning records and enabling selective disclosure. Leveraging secure enclaves ensures the confidentiality and integrity of credential data, substantially improving both privacy and usability. Table~\ref{tab:comparison} summarizes these distinctions between our proposed framework and existing LER solutions.

\begin{table*}[htbp]
    \centering
    \caption{Comparison of Credentialing and LER Frameworks}
    \resizebox{\textwidth}{!}{%
    \begin{tabular}{lcccc}
        \hline
        \textbf{Framework} & \textbf{Blockchain Anchored} & \textbf{AI Analytics} & \textbf{Informal Records} & \textbf{Self-Sovereign Wallet} \\
        \hline
        LegiCred \cite{legicred_digital_transformation} & Yes & Issuer & No & Partial \\
        OpenCreds (LRNGC) \cite{opencreds_framework} & Yes & Issuer & No & Partial \\
        Hyperledger Indy \cite{hyperledger_indy} & Yes & Both (SSI) & No & Yes \\
        Open Wallet Foundation \cite{open_wallet} & No & N/A & No & Yes \\
        DCC \cite{dcc_framework} & Yes & Issuer & No & Yes \\
        \hline
        \textbf{Our Approach} & No & Issuer or Self & \textbf{Yes} & \textbf{Yes} \\
        \hline
    \end{tabular}}

    \vspace{1ex}
    \begin{flushleft}
        \footnotesize{\textit{Note:} ``Blockchain Anchored'' indicates use of blockchain or distributed ledger for credential immutability. ``AI Analytics'' indicates the use of skill extraction or advanced analytics. ``Nontraditional Records'' refers to credentials from Massive Open Online Courses (MOOCs) \cite{mooc2025} or informal learning platforms. ``Self-Sovereign Wallet'' denotes support for user-controlled identity management and credential storage.}
    \end{flushleft}
    \label{tab:comparison}
\end{table*}

\noindent\textbf{External validation and domain coverage}
The syllabus--to--O*NET pipeline adopted here follows Course--Skill Atlas~\cite{javadian2024course}, developed and stress-tested at national scale (3.16\,M U.S.\ syllabi across 62 fields of study). That work documents cleaning and mapping transparently and reports macro validations,including aggressive filtering of non-pedagogical sentences, stability under sub-sampling, and auxiliary regressions linking detailed work activities(DWAs) to O*NET abilities, which collectively support face validity in lieu of gold labels. The methodology and its validation posture are adopted here.

\section{Background}
\label{build}

This section introduces the decentralized identity and credential primitives used in the framework and the skills taxonomy referenced by the NLP pipeline.

\subsection{Decentralized Identifiers (DIDs)}
\label{did}

Decentralized Identifiers (DIDs) are URIs that enable verifiable, self-sovereign identifiers without reliance on a centralized identity provider (IdP) or certificate authority (CA)~\cite{DIDs}. Each DID resolves, via its method, to a DID Document that specifies verification methods (public keys and supported proof mechanisms) and optional service endpoints. Control of a DID is proven by possession of the corresponding private key(s) referenced in the DID document; the DID itself is not an IdP or CA.

For our purposes it is sufficient to expose the following abstract interfaces:
\[
\begin{aligned}
\mathsf{GenDID}(pk_c,\mathit{mtd},\mathit{meta}) &\rightarrow \mathit{did},\\
\mathsf{Resolve}(\mathit{did}) &\rightarrow \mathit{doc},\\
\mathsf{ProveCtrl}(sk_c,\mathit{challenge}) &\rightarrow \sigma,
\end{aligned}
\]
where $pk_c/sk_c$ denote the controller key pair, $\mathit{mtd}$ the DID method, and $\mathit{doc}$ the resolved DID document containing verification methods and metadata. $\mathsf{ProveCtrl}$ produces a proof of control $\sigma$ (e.g., a signature over a verifier challenge).

\subsection{Verifiable Credentials (VCs): Institutional and Self-Issued}
\label{vrfcred}

Verifiable Credentials (VCs) are tamper-evident, digitally signed claims issued about a subject~\cite{vc_standard}. Selective-disclosure presentations allow holders to reveal only the attributes required by a verifier. We model issuance, presentation, verification, and status as:
\[
\begin{aligned}
\mathsf{Issue}(sk_I,\mathit{claims}) &\rightarrow VC,\\
\mathsf{Present}(VC,\mathsf{pol},sk_H) &\rightarrow \mathit{VP},\\
\mathsf{Verify}(vk_I,\mathit{VP}) &\rightarrow \{0,1\},\\
\mathsf{Status}(\mathit{VCID}) &\rightarrow \mathit{valid/revoked}.
\end{aligned}
\]
Here $sk_I/vk_I$ are issuer keys, $sk_H$ is the holder key, $\mathsf{pol}$ is a disclosure policy, and $\mathit{VP}$ is a holder-signed verifiable presentation.

Two credential classes are used. \textbf{Institution-issued} VCs (e.g., transcripts, certificates) carry signatures from authoritative issuers such as universities, providing high assurance. \textbf{Self-issued} VCs capture informal learning and are signed by the holder; in the proposed system their derivation is additionally attested by the TEE, strengthening verifiability without disclosing raw artifacts.

\begin{figure}[htbp]
    \centering
    \includegraphics[width=1\linewidth]{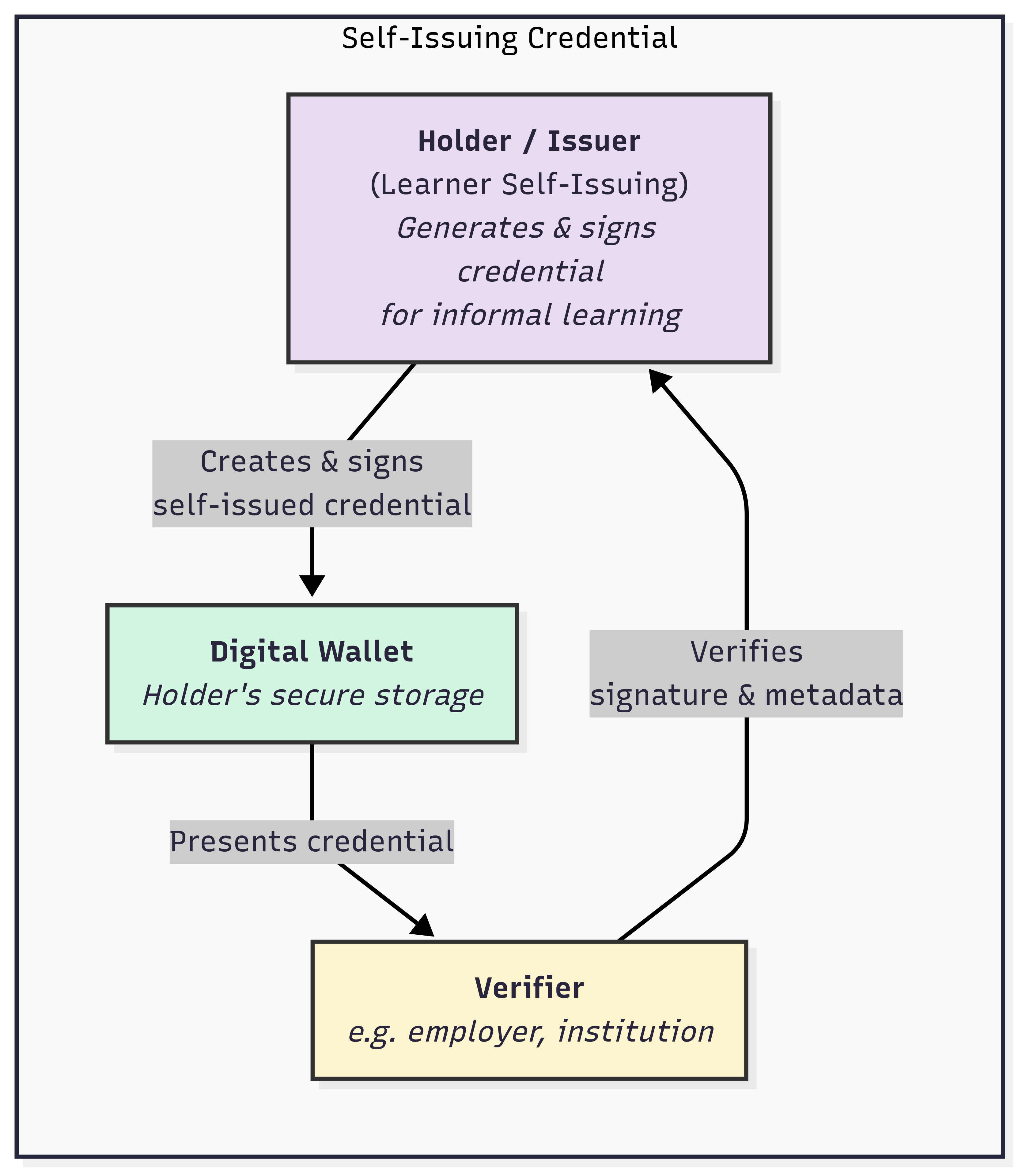}
    \caption{Self-issued credential flow. A holder derives and signs a credential for informal learning; the derivation is attested by the enclave and the credential is stored in the wallet for selective-disclosure presentations.}
    \label{fig:self-issued-flow}
\end{figure}

\subsection{NLP-Based Skill Extraction with O*NET}
\label{sec:nlp-onet}

To obtain standardized and interpretable skill evidence, the NLP pipeline aligns unstructured educational records with the O*NET taxonomy maintained by the U.S. Department of Labor~\cite{onet2025taxonomy}. O*NET provides: (i) \emph{Detailed Work Activities (DWAs)}, fine-grained task descriptors (e.g., ``Develop scientific or mathematical models''); (ii) \emph{Task statements}, higher-level responsibilities composed of multiple DWAs; and (iii) \emph{Abilities}, innate or acquired capacities (e.g., \emph{Deductive Reasoning}) relevant across occupations.

Following the Course--Skill Atlas methodology, sentence embeddings are computed for learning records and compared to O*NET descriptors. Each skill $s$ receives a score given by the maximum cosine similarity between $s$ and any sentence in the document:
\begin{equation}
\mathrm{Score}(s) \;=\; \max_{j}\, \operatorname{cos}\!\bigl(\vec{v}_{\mathrm{sent}_j},\, \vec{v}_{s}\bigr),
\end{equation}
where $\vec{v}_{\mathrm{sent}_j}$ is the embedding of the $j$-th sentence and $\vec{v}_s$ is the embedding of the skill descriptor. The resulting skill vector underlies the derivative skill credential, aligning learner records with labor-market terminology and enabling privacy-preserving matching against job requirements.

\section{System Setting}
\label{setting}

The proposed system enables privacy-preserving, AI-supported management of LER through a decentralized architecture. Verifiable credentials, trusted execution environments (secure enclaves), and an NLP pipeline are integrated to validate formal and informal learning, derive structured skill evidence, and perform job-skill matching under holder control.  Figure~\ref{fig:credential_framework} illustrates the architecture, highlighting interactions among key actors and components.

\begin{figure*}[!t]
  \centering
\includegraphics[width=\textwidth,height=0.80\textheight,keepaspectratio]{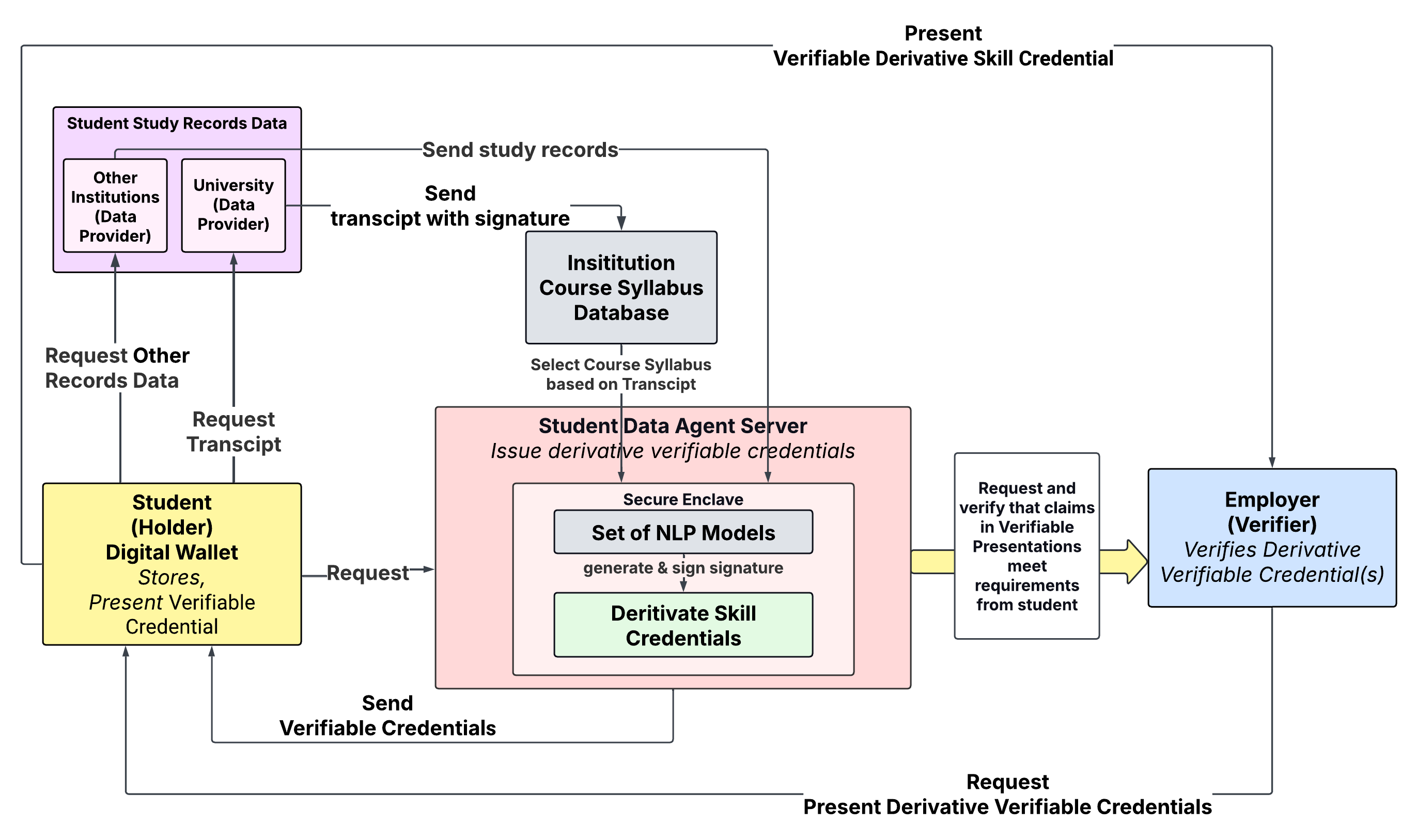}
  \caption{System architecture of the decentralized LER. Institution-issued credentials (e.g., transcripts) are delivered to the holder’s digital wallet and used to select relevant syllabi from the institutional database; optional study records from other providers may also be ingested. Within the Student Data Agent server, a secure enclave executes the NLP pipeline to derive and sign skill credentials, which are stored alongside issuer credentials in the wallet. Upon verifier request, derivative credentials are assembled into verifiable presentations and checked against stated requirements, while raw records remain undisclosed.}
  \label{fig:credential_framework}
\end{figure*}

\subsection{Entity Roles}
Three principal roles participate in issuance, control, and verification of credentials:
\begin{itemize} \item \textbf{Issuer} ($\mathcal{I}$): An authoritative entity that signs and attests to credentials (e.g., universities, certification bodies). Institutional credentials originate from authoritative sources (e.g., transcripts, certificates). Self-issued credentials for informal learning are also supported; their derivations are attested by the enclave to strengthen verifiability while preserving confidentiality (see Fig.~\ref{fig:credential_framework}).

\item \textbf{Holder} ($\mathcal{H}$): The learner who stores and manages credentials in a digital wallet, exercises selective disclosure according to verifier requests and policy, and authorizes presentations. Both institutional and self-issued credentials are controlled by the holder.

\item \textbf{Verifier} ($\mathcal{V}$): A relying party (e.g., employer, institution) that requests evidence and validates it. Verification checks issuer signatures for institutional credentials and enclave attestation and provenance bindings for derivative/self-issued skill credentials, without requiring access to raw records.
\end{itemize}

This delineation of roles ensures that credentials remain verifiable, holder-controlled, and privately validated: issuance is authenticated by \(\mathcal{I}\), storage and disclosure are governed by \(\mathcal{H}\), and acceptance rules are enforced by the verifier ($\mathcal{V}$) via signature and status checks, freshness/nonce validation, an allowlist of enclave measurements, equality of policy and matching thresholds.

\subsection{Threat Model}
Protected assets comprise raw educational records $D$ (e.g., transcripts, syllabi), derived skill vectors $\mathbf{v}_H$, issued verifiable credentials $VC$, the holder’s private keys $sk_H$, and enclave code and data. Trust anchors consist of issuer public keys $pk_I$ used to validate institutional credentials and the TEE’s remote-attestation keys. Objectives are to preserve the confidentiality of $D$ and $sk_H$, to ensure the integrity and authenticity of issued $VC$ and matching outputs, and to prevent linkage across multiple presentations beyond attributes intentionally disclosed by the holder.

Adversarial capabilities include a network-level attacker able to eavesdrop, replay, reorder, inject, or drop messages, and a host-level adversary with control over the operating system, hypervisor, and cloud host. In this setting, non-enclave memory, storage, scheduling, and I/O may be inspected or modified, and denial of service is possible. Verifiers are modeled as honest-but-curious: protocol steps are followed, yet adaptive queries may be issued to infer additional information or to correlate multiple sessions. Wallet and storage operators are likewise treated as honest-but-curious, with access to metadata but not to enclave-sealed plaintext. A malicious holder may submit arbitrary evidence in an attempt to obtain false self-issued credentials; integrity is constrained by enclave attestation and by issuer signatures for institutional credentials.

Standard cryptographic primitives and TEE isolation with correct remote attestation are assumed secure. Physical attacks and microarchitectural side channels against TEEs are considered out of scope; commonplace mitigations (e.g., constant-time operations and minimal state exposure) are applied but not modeled. Issuers are trusted to sign correct credentials and to manage keys and revocation; systemic mis-issuance or compromise of the issuer public-key infrastructure is not addressed. Upstream institutional records (transcripts and syllabi) are assumed authentic at the source; poisoning of those sources lies outside the present scope. Holder devices are assumed uncompromised at the time of use, with private keys stored in hardware-backed keystores. Denial-of-service resilience and full traffic-analysis resistance are not primary goals of this design.

Under these assumptions, the security objectives are as follows. Confidentiality requires that raw inputs $D$ and $sk_H$ remain confined to the enclave and that verifiers learn only attributes or predicates explicitly disclosed by the holder. Integrity and authenticity require that derived credentials and matching outputs be bound to attested enclave code and to specific inputs via digital signatures and attested transcripts. Unlinkability requires that repeated presentations be unlinkable except through attributes deliberately revealed or holder-chosen identifiers. Verifiability and revocation require that verifiers be able to check issuer signatures, TEE attestation evidence, and credential status; policy and mechanisms for revocation are enabled by design but are not analyzed in detail here.


\section{Proposed Framework}
\label{proposed}

This section describes our privacy-preserving, AI-enabled decentralized LER framework, detailing its high-level architecture, entity roles, and comprehensive workflow for credential issuance and verification.

\subsection{System Architecture Overview}
Our proposed LER framework combines verifiable credentials, secure enclaves, and NLP techniques to facilitate secure, privacy-preserving verification of learning records, automated extraction of skills, and efficient job-skill matching. Users retain complete control over their formal and informal credentials, selectively sharing only essential information with verifiers to safeguard privacy and reduce potential bias. 

Employers benefit from streamlined credential validation, significantly reducing verification time, mitigating fraud risks, and efficiently screening candidates based on verified skills. The lightweight and decentralized architecture of the system, illustrated in Figure~\ref{fig:credential_framework}, ensures secure interactions between users, institutions, and employers, effectively protecting the confidentiality, integrity, and individual autonomy of data throughout the lifecycle of credential management.

\subsection{System Workflow with Credential Issuance \& Verification}
The end-to-end process proceeds in six steps. Steps~1--3 cover credential issuance; Steps~4--6 cover presentation, verification, and matching.

\subsubsection{Step 1: DID setup and institutional issuance}
A decentralized identifier for the holder is generated as $\mathrm{DID}_H \leftarrow \mathsf{GenDID}(pk_H,\mathit{mtd},\mathit{meta})$. An authorized issuer $\mathcal{I}$ produces an institutional credential over record $D$ with signature $\sigma_D \leftarrow \mathsf{Sign}(sk_I,D)$ and delivers $(D,\sigma_D)$ to the holder’s wallet.

\subsubsection{Step 2: Evidence collation and enclave derivation (self-issued)}
Supported informal evidence is supplied to a holder-controlled enclave (TEE). Inside the enclave, an NLP pipeline analyzes the evidence and produces a structured skill vector $\hat v$. A self-issued skill credential $VC_{\text{skill}}$ denoted $VC$ is created that embeds $\hat v$ and provenance fields; the enclave outputs $VC$ together with a hardware attestation proving that derivation occurred inside a measured enclave.

\subsubsection{Step 3: Wallet storage and selective disclosure control}
Institutional $(D,\sigma_D)$ and self-issued $VC$ are stored in the holder’s wallet. A disclosure policy is configured to govern which attributes or predicates may be revealed in future presentations.

\subsubsection{Step 4: Presentation (holder)}

In response to a verifier request, a verifiable presentation is constructed; only attributes permitted by the disclosure policy are included. Let
\begin{equation*}
  \mathit{VP} \gets \mathsf{Present}(VC,\mathsf{pol},sk_H),
\end{equation*}
optionally with stapled revocation status. Private keys and raw records remain confined to the enclave.

\subsubsection{Step 5: Verification (verifier enclave)}
The verifier $\mathcal{V}$ validates the presentation inside a verifier-side enclave. Institutional credentials are checked via issuer signatures and status. For self-issued credentials, the enclave verifies the hardware attestation (measurement, input commitments, freshness) and checks provenance consistency; the holder signature on $\mathit{VP}$ and revocation status within a freshness window $\Delta$ are enforced.

\subsubsection{Step 6: Skills-only matching (verifier enclave)}
Matching is performed over attested skill vectors only. Let $s=f(\hat v)$ and $\hat Y=\mathbb{1}\{s\ge\tau\}$. Non-skill resume  fields are not inputs, rendering outcomes invariant to those fields. Because matching executes inside the enclave, job criteria and holder credentials remain confidential; only the decision or a minimal score is released according to policy.

The workflow yields verifiable, selectively disclosable credentials and decisions whose provenance is bound to measured code and specific inputs. Raw records and private keys never leave the enclave, and repeated presentations remain unlinkable beyond intentional disclosure.

\subsection{System Components}
\label{component}

The framework comprises three components: a digital wallet, an NLP processing pipeline, and secure enclaves. Together they enable secure credential storage, automated skill extraction, and privacy-preserving verification.

\subsubsection{Digital Wallet} The wallet stores and manages the holder’s verifiable credentials (institutional and self-issued), enforces key protection, and supports selective disclosure so that only required attributes are shared with verifiers.

\subsubsection{NLP Processing Pipeline}
\label{nlp}
An NLP pipeline extracts structured skill profiles from transcripts, syllabi, and supported informal artifacts. Following the Syllabus-to-O*NET approach of Course--Skill Atlas, the pipeline (i) filters non-pedagogical text, (ii) embeds sentences using the all-mpnet-base-v2 Sentence-BERT model \cite{reimers2019sentencebert,allmpnetbasev2}, (iii) maps course content to O*NET DWAs and task statements by semantic similarity, and (iv) personalizes scores using student-specific grade and course-level weights.

\noindent \textbf{Set-Weight Factors for Personalized Skill Extraction} We introduce a set-weight step integrating transcript data to produce individualized skill credentials. This step adapts skills scores based on the student's grade performance, course level, and other optional factors.

Each course $c$ yields a vector $\mathbf{v}_c \in \mathbb{R}^{m}$ over a target skill set $S=\{s_1,\ldots,s_m\}$ (e.g., O*NET DWAs). If a syllabus contains $n$ outcome sentences, then
\begin{equation}
\mathbf{v}_c = \bigl[\, \max_{1\le j\le n}\;\mathrm{sim}\!\bigl(\mathbf{v}_{\mathrm{sent}_j},\,\mathbf{v}_{s_i}\bigr) \,\bigr]_{i=1}^{m},
\end{equation}
where $\mathbf{v}_{\mathrm{sent}_j}$ is the embedding of sentence $j$ and $\mathbf{v}_{s_i}$ the embedding of skill $s_i$.
To obtain a student-specific skill vector, grade and level weights are applied:
\begin{equation}
\mathbf{v}_H \;=\; \sum_{c \in C_H} \mathbf{v}_c \cdot w_{\mathrm{grd}}(c) \cdot w_{\mathrm{lvl}}(c),
\end{equation}
where $C_H$ is the student’s completed-course set, $w_{\mathrm{grd}}(c)$ reflects performance in $c$, and $w_{\mathrm{lvl}}(c)$ reflects course rigor. The result is a personalized, standardized profile aligned to O*NET terminology.

\subsubsection{Secure Enclave (SE) Integration}

Trusted execution environments (TEEs) provide hardware-enforced isolation of code and data~\cite{costan2016sgx}. In the framework, sensitive operations, NLP skill extraction, provenance construction, and issuance of derivative skill credentials are executed entirely inside an enclave. Raw transcripts, intermediate NLP outputs, and private keys remain confined to enclave memory. Remote attestation allows verifiers to confirm that (i) inputs remained confidential, (ii) measured code executed untampered, and (iii) outputs are bound to that execution.

\begin{table}[t]
\caption{Attestation evidence bound to outputs.}
\label{tab:attestation}
\centering
\renewcommand{\arraystretch}{1.1}
\small
\begin{tabularx}{\linewidth}{l >{\raggedright\arraybackslash}X}
\toprule
\textbf{Field} & \textbf{Description} \\
\midrule
$M_e$            & Enclave measurement (code+model hash) \\
$H_{\text{inputs}}$ & Commitment to inputs (salted hashes of transcript and syllabus/LO) \\
$H_{\text{policy}}$ & Hash of disclosure/acceptance/matching configuration \\
$n_V$            & Verifier challenge nonce (freshness) \\
$t$              & Enclave timestamp / monotonic counter \\
$\sigma_{\text{TEE}}$ & Hardware attestation over all fields \\
$H_{\text{prov}}$ & Provenance digest embedded in $VC_{\text{skill}}$ \\
\bottomrule
\end{tabularx}
\end{table}

The verifier receives hardware attestation and a holder-signed presentation binding code, inputs, and policy to the credential. Let $M_e$ be the enclave measurement (hash of enclave binary and model bundle). The salted input commitment is
\[
H_{\text{inputs}} =
  H\!\big(
    \texttt{"inputs"} \,\|\, s \,\|\, H(\text{transcript}) \,\|\, H(\text{syllabus/LO})
  \big),
\]
where $H(\cdot)$ is a collision-resistant hash, $\|$ denotes concatenation, and $s$ is a per-presentation salt. Let $H_{\text{policy}} = H(\mathrm{canon}(\mathsf{pol}))$, with verifier nonce $n_V$ and timestamp $t$. The attestation and provenance are:
\[
\begin{aligned}
\sigma_{\text{TEE}} &=
  \mathrm{Attest}_{K_{\text{TEE}}}\big(
    \texttt{"att"} \,\|\, M_e \,\|\, H_{\text{inputs}} \,\|\, H_{\text{policy}} \,\|\, n_V \,\|\, t
  \big),\\
H_{\text{prov}} &=
  H\big(
    \texttt{"prov"} \,\|\, M_e \,\|\, H_{\text{inputs}} \,\|\, H_{\text{policy}} \,\|\, t
  \big).
\end{aligned}
\]
Verification passes if issuer/holder signatures and $\sigma_{\text{TEE}}$ verify, revocation is fresh (within $\Delta$), $M_e$ is allowlisted, and recomputed $H_{\text{prov}}$ and $H_{\text{policy}}$ match.

\subsection{Trusted Computing Base (TCB).}
The TCB comprises: (i) the enclave binary (normalizer, embedding and scoring, credential issuer, policy engine) and linked cryptographic libraries; (ii) model parameters packaged with the binary; (iii) the TEE hardware/firmware and attestation keys; and (iv) configuration bound via $H_{\text{policy}}$ and a trusted time/nonce source. The host OS/hypervisor, network, storage, wallet UI, and issuer resolver/adapters are outside the TCB. Remote attestation binds $M_e$, $H_{\text{inputs}}$, and $H_{\text{policy}}$ to the presentation, so outputs are accepted only from the measured TCB.

Issuer DIDs are resolved to syllabus/learning-outcome endpoints. Signed artifacts are fetched outside the enclave; signature verification and schema normalization are executed inside the enclave to minimize reliance on external components. Only hashes of artifacts ($H_{\text{inputs}}$) are retained; raw documents never leave the enclave.

\paragraph*{Provenance and Dispute Handling}
\raggedright
Each $VC_{\text{skill}}$ includes:
prov.inputs$=[H(\text{transcript}),\,H(\text{syllabus})]$,
prov.code$=M_e$,
prov.policy$=H_{\text{policy}}$,
prov.time$=t$, and \texttt{derivation\_id}.
\par
\justifying
If contested, a corrected credential is issued; the old one is revoked and its revoked status remains visible.

\subsection{Revocation and Freshness}
Institutional credentials follow issuer-hosted status lists. Derivative skill credentials are short-lived and are additionally referenced by a \texttt{DerivativeStatusList}. Verifiers require stapled status within a freshness window $\Delta$ and reject stale presentations. DID/key rotation is supported: new enclave measurements $M_e$ are associated with new issuance keys, and prior keys are marked revoked in status lists. Disputes trigger re-issue with corrected inputs; superseded credentials are revoked.

\section{Bias Reduction by Design}
This section formalizes that the pipeline reduces opportunities for bias at screening by enforcing skills-only matching. Let $v\in\mathbb{R}^{m}$ denote the attested skill vector produced within the enclave from signed records and normalized syllabi, and let $z$ collect non-skill r\'esum\'e fields (e.g., name, institution, dates). The verifier receives only disclosures derived from $v$ under selective disclosure; personal or demographic attributes are never inputs to matching.

\begin{definition}[Skill-only matcher]
A matcher is skill-only if the score is $s=f(v)$ and the decision is
$\hat{Y}=\mathbb{1}\{s\ge \tau\}$ for threshold $\tau$, with no dependence on $z$.
\end{definition}

\begin{definition}[Non-skill invariance]
A matcher is non-skill invariant if, for any fixed $v$ and any $z,z'$,
$s(v,z)=s(v,z')$ and $\hat{Y}(v,z)=\hat{Y}(v,z')$.
\end{definition}

\begin{proposition}[Attribute-blindness]
\label{prop:blind}
Every skill-only matcher is non-skill invariant.
\end{proposition}

\begin{proof}
By definition, $s=f(v)$ depends only on $v$; therefore, holding $v$ fixed while varying $z$ cannot change $s$ or $\hat{Y}$.
\end{proof}

\begin{definition}[Bias–Opportunity Index (BOI)]
\label{def:boi}
For a matcher $h(v,z)$, define
\[
\mathrm{BOI}(h) \;=\; \mathbb{E}_{v}\, \mathbb{E}_{z,z'}\!\big[(\,h(v,z)-h(v,z')\,)^{2}\big].
\]
The index measures sensitivity of decisions/scores to non-skill edits when skills are held fixed.
\end{definition}

\begin{proposition}[Reduced bias opportunity]
\label{prop:boi}
For any skill-only matcher $f(v)$, $\mathrm{BOI}(f)=0$.
For any r\'esum\'e-based matcher $g(v,z)$ that depends on $z$ on a set of non-zero measure,
$\mathrm{BOI}(g)>0$.
\end{proposition}

\begin{proof}
For skill-only $f$, $f(v,z)=f(v,z')$ for all $z,z'$, so the squared difference is zero pointwise and the expectation is zero. If $g$ depends on $z$ on a set of non-zero measure, then the squared difference is strictly positive on that set, yielding a positive expectation.
\end{proof}

\begin{corollary}[Zero flip probability under non-skill edits]
\label{cor:flip}
Let $\hat{Y}_h(v,z)$ denote the decision of matcher $h$. For skill-only $f$,
\[
\Pr_{z,z'}\!\big[\hat{Y}_f(v,z)\neq \hat{Y}_f(v,z')\big]=0
\quad\text{for all } v .
\]
\end{corollary}

\begin{proof}
Immediate from non-skill invariance (Proposition~\ref{prop:blind}).
\end{proof}
\begin{proposition}[Information-leakage bound]
\label{prop:mi}
Let $A$ denote any sensitive attribute not used by the matcher. With the Markov chain $A\!\to\! v\!\to\! s=f(v)$, the data-processing inequality gives
$I(A;s)\le I(A;v)$. In particular, if $v$ contains no information about $A$ (e.g., by construction or filtering), then $I(A;s)=0$.
\end{proposition}

\paragraph*{System implications}
The enclave and disclosure policy ensure that matching is executed on $v$ alone; personal and demographic fields $z$ are neither processed nor revealed. Propositions~\ref{prop:boi}–\ref{cor:flip} therefore hold by construction for the proposed verifier-side matching, demonstrating that screening-time sensitivity to non-skill heuristics is eliminated. Proposition~\ref{prop:mi} further shows that any residual information about sensitive attributes present in $v$ cannot increase through the matching function.

\paragraph*{Limitation}
These guarantees address bias introduced at screening time by non-skill fields. They do not correct upstream differences in the distribution of skills across populations; group-level parity depends on those distributions and is out of scope for this work.

\section {Security and Privacy Analysis}
\label{analysis}

The proposed LER framework ensures security and privacy, assuming the integrity of digital signatures, the security of verifiable credentials, and correct enclave attestation. Under these conditions, adversaries cannot forge credentials, transfer them between users, or access sensitive transcript data during NLP skill extraction, as all processing occurs within secure enclave memory. Holders maintain complete autonomy over credential disclosure, selectively revealing only necessary attributes (e.g., course name or skill level) to verifiers, consistent with self-sovereign identity principles. The NLP pipeline enhances credentials derived from raw sources (e.g., resumes, online transcripts) by identifying and embedding structured skill descriptors, allowing learners to protect sensitive personal information during credential verification.

We define formal security theorems for unforgettability and data confidentiality within the proposed LER framework as follows:
\begin{theorem} \textbf{Privacy.} If there exists a PPT adversary $\mathcal{A}$ that can break the confidentiality of transcript data $D$ processed within the secure enclave with non-negligible probability $\epsilon$, we construct a PPT challenger $\mathcal{C}$ that breaks the security of the secure enclave with the same probability $\epsilon$.

\end{theorem}

\begin{proof} $\mathcal{C}$ is given the public key and oracle access to secure enclaves $\mathcal{O}^{SE}$. It selects a document $D$ and submits it to $\mathcal{A}$ and receives a verifiable credential $VC$. Then, $\mathcal{C}$ prepares two different documents $D_0$ and $D_1$, randomly chooses $b \in \{0, 1\}$ and using $\mathcal{O}^{SE}$, receives ${VC}_b$ on the chosen $D_b$, which it provides to $\mathcal{A}$. If the adversary $\mathcal{A}$ successfully guesses the bit $b$, the $\mathcal{C}$ outputs the same bit $b$.

If $\mathcal{A}$ can distinguish between ${VC}_0$ and ${VC}_1$ with a non-negligible probability, it means there is information leakage in the output ${VC}_b$, indicating failure in the enclave's attestation mechanism, contradicting our security assumption.
\end{proof}
Our proposed framework offers a streamlined approach to capturing and verifying learners' achievements by uniting the different contents under a consistent workflow. Students gain fine-grained control over their records, and verifiers benefit from dependable, up-to-date skill data. 

\begin{theorem} \textbf{Unforgeability.}  Assume there is a probabilistic polynomial time (PPT) adversary $\mathcal{A}$ who can produce a valid credential on document $D$ with non-negligible probability $\epsilon$. We can build a challenger $\mathcal{C}$ to break at least one of the security assumptions, including unforgeability of the signature scheme, security of $VC$, or security of the enclaves attestation mechanism with the same probability $\epsilon$.
\end{theorem}

\begin{proof} The challenger $\mathcal{C}$ is given the public key of the issuer and has access to the signing oracle $\mathcal{O}^{Sig}$. The $\mathcal{C}$ forwards the public key to $\mathcal{A}$. When $\mathcal{A}$ requests a signature on document $D_i$, the $\mathcal{C}$ forwards this request to $\mathcal{O}^{Sig}$ and returns the signature $\sigma_i$ to $\mathcal{A}$. Finally, $\mathcal{A}$ outputs a forged credential $VC^\prime$ for document $D^\prime$. If $\mathcal{C}$ is satisfied, i.e., $\mathrm{VCVrf}(vk, D^\prime, y^\prime, \pi^\prime) = 1$, where $(y^\prime, \pi^\prime)$ are extracted from $VC^\prime$, on an authorized document $D^\prime$ that was not previously queried to the signing oracle, it means $\mathcal{A}$ could successfully forge the LER credential.
 
Since any signature generated on $D^\prime$ must be verified before processing by the secure enclave that generated ${VC}^\prime$, a valid ${VC}^\prime$ must be obtained from a valid signature $\sigma^\prime$ on $D^\prime$. Hence, $\mathcal{C}$ breaks the unforgeability of the digital signature, where it can extract this signature and output $(D^\prime, \sigma^\prime)$ as its forgery, contradicting our security assumption.
\end{proof}

\section{Experimental Evaluation}
\label{evaluation}

This section describes the experimental setup and the evaluation protocol used to assess effectiveness and system performance.

\subsection{Experimental Setup}
The NLP skill–extraction pipeline was implemented in Python and containerized with Docker for reproducibility. Experiments were executed on a local workstation (Intel Core i7, 16\,GiB RAM, Ubuntu~22.04) and on an AWS Nitro Enclave configured with comparable resources. Inputs (transcripts and r\'esum\'es) were processed in batches of 5--100 items, both sequentially and with controlled parallelism. For cloud runs, system counters and wall-clock times were exported to Amazon CloudWatch. All enclave runs used a single measured binary; raw artifacts and private keys remained enclave-confined.

\subsection{Artifacts and Tasks}
A representative student skill profile was derived from computer-science syllabi collected from an anonymous R1 university. Two public job descriptions were selected: (i) a Java-developer role and (ii) a C\#/data-mining role. Job texts were normalized (tokenization, sentence segmentation, stop-word filtering) and embedded using the same model family as the syllabus-to-O*NET mapping to maintain embedding-space compatibility.

\subsection{Metrics}
\paragraph*{Skill-matching effectiveness.}
(i) \emph{Binary overlap} between attested skills and required skills:
\[
\mathrm{Overlap@}k
 \;=\;
 \frac{|\mathrm{Top}\,k(v)\cap R|}{|R|},
\]
where $v$ is the attested skill vector, $\mathrm{Top}\,k(v)$ the top-$k$ skills by score, and $R$ the set of skills extracted from the job description.
(ii) \emph{Semantic similarity} as the mean maximum cosine similarity between job skills and candidate skills:
\[
\mathrm{SemSim}
 \;=\;
 \frac{1}{|R|}\sum_{s\in R}\max_{s'\in S}\cos\!\big(\vec v_s,\vec v_{s'}\big),
\]
with $S$ the candidate’s skill set and $\vec v_\cdot$ the embeddings.

\paragraph*{System performance.}
Metrics include: (i) \emph{per-document latency} for ingestion and extraction ($p50/p95$), 
(ii) end-to-end \emph{presentation–verify–match} latency ($p50/p95$) within the verifier enclave, 
(iii) \emph{CPU time} and \emph{peak RAM} for enclave and non-enclave segments, and 
(iv) \emph{I/O overhead} for status-stapling and attestation verification. 
Nonparametric bootstrapping is applied to compute confidence intervals.

\subsection{Evaluation Methodology}
Tests were run in both environments on transcript datasets (1{,}000--3{,}000 words each).
Two dimensions were varied:
\begin{enumerate}
  \item \textbf{Batch size:} 5, 10, 20, 30, and 40 items (Fig.~\ref{fig:execution-time-small}), and 50, 60, 70, 90, and 100 items (Fig.~\ref{fig:execution-time-large}).
  \item \textbf{Container and host utilization:} Docker CPU (\%), memory (MiB), active container processes, and host-level utilization.
\end{enumerate}
Each experiment was repeated three times and averaged to reduce noise. Pipeline configuration, the embedding model, and pre-processing were held constant to isolate infrastructure effects.

\smallskip
\noindent\textbf{Performance and scalability.}
On a local workstation (Intel Core i7, 16\,GiB RAM), the matching stage consistently executed in approximately $0.1$\,s per job description, utilized under $10\%$ of a single CPU core, and consumed under 100\,MiB of memory, indicating a lightweight implementation suitable for real-time use. In the system design, matching is executed inside the verifier enclave; only the decision or a minimal score is disclosed according to policy.

\subsection{Skill Extraction Accuracy and External Validation}
\label{sec:accuracy}
The NLP mapping follows the Syllabus--to--O*NET method of Course--Skill Atlas. That study reports: (i) large-scale domain coverage (3.16\,M syllabi across 62 fields of study); (ii) content filtering that retains only pedagogical sentences (approximately 14\% of raw sentences after cleaning); (iii) stability under sub-sampling with an ``elbow'' near nine syllabi per institution$\times$major$\times$year; and (iv) auxiliary ability-prediction regressions with mean squared error below 0.025 for most mappings. As gold labels for sentence$\to$DWA (Detailed Work Activity) assignments do not exist at scale, these macro validations are used in lieu of a gold-label evaluation.

In the prototype, the same pipeline and similarity scoring are executed inside the enclave. A repeatability check on the local corpus, under identical inputs and configuration, yielded $<\!5\%$ variance among the top-ranked skills across repeated runs, indicating stable derivations under fixed conditions. No new gold-label annotation is introduced; instead, the externally validated methodology and the observed stability are relied upon to justify mapping quality for the intended systems use.

\noindent\textbf{Scope and limitations.}
Inputs predominantly comprise computer-science syllabi and representative informal artifacts; extending coverage across additional disciplines and assessing robustness to noisy or partial documents are left to future work. Results should therefore be interpreted as evidence of feasibility under the stated setting rather than a cross-discipline benchmark.

\noindent \textbf{Operational performance.}
The end-to-end path combines (i) presentation and nonce exchange, (ii) attestation verification, and (iii) skills-only matching. The dominant compute is the NLP pipeline. The matching stage was observed at approximately $0.1$\,s per job description, and attestation verification time was dominated by signature checks; no accumulators beyond signature verification were required.

\subsection{Skill-Matching Methodology}
\noindent\textbf{Binary overlap.}
Direct overlap was measured between the candidate’s extracted skill set $S_{\text{student}}$ and the skills explicitly required in the job description $S_{\text{job}}$:
\begin{equation}
\mathrm{Overlap} \;=\; \frac{\lvert S_{\text{student}} \cap S_{\text{job}}\rvert}{\lvert S_{\text{job}}\rvert}.
\end{equation}
For the Java–developer role, an $80\%$ overlap (8/10) was observed, covering \emph{Java}, object–oriented programming, data structures, algorithms, SQL, software engineering, version control (Git), and operating systems. For the C\# data–mining role, a $70\%$ overlap (7/10) was observed, covering object–oriented programming, data structures, algorithms, SQL, data mining, machine learning, and Python. Lower overlaps typically reflect specialized, industry–specific competencies (e.g., .NET frameworks).

\smallskip
\noindent\textbf{Hybrid semantic matching.}
To account for near matches and synonyms, semantic similarity was incorporated using vector embeddings. Cosine similarity between two skill embeddings $\mathbf{v}_i$ and $\mathbf{v}_j$ is
\begin{equation}
\mathrm{cos}(\mathbf{v}_i,\mathbf{v}_j) \;=\; \frac{\mathbf{v}_i \cdot \mathbf{v}_j}{\lVert \mathbf{v}_i\rVert \,\lVert \mathbf{v}_j\rVert }.
\end{equation}
Each job-required skill was compared against all candidate skills, retaining the maximum similarity per job skill.
For the Java–developer role, the strongest matches (similarity) included: Java (1.00), Data Structures (1.00), Algorithms (1.00), Object–Oriented Programming (0.99), SQL (0.99), Software Engineering (0.98), Operating Systems (0.95), Version Control (Git) (0.92), Python (0.85), and C++ (0.80).
For the C\# data–mining role, the top matches included: Object–Oriented Programming (1.00), Data Structures (1.00), Algorithms (1.00), SQL (1.00), Data Mining (1.00), Machine Learning (1.00), and Python (1.00), with related skills such as Java (0.88), C++ (0.85), and Software Engineering (0.80) also identified.

\section{Results}
Tables~\ref{tab:system-metrics-comparison} compare system metrics for typical batch runs on a local machine versus an AWS Nitro Enclave. On the local workstation, Docker CPU usage exceeds \(1000\%\), indicating extensive multi-core utilization; memory consumption averages \(\approx 1551\)\,MB, and the host load average settles near \(8.91\). By contrast, the AWS Nitro Enclave maintains Docker CPU usage near \(102\%\), with overall host CPU usage at \(51.81\%\). Although Docker memory usage decreases slightly to \(\approx 1147\)\,MB, total host memory rises to \(\approx 9529\)\,MB out of \(16\)\,GiB available. The AWS load average remains near \(1.03\), suggesting capacity headroom for additional workloads. Collectively, these snapshots indicate that the local system approaches saturation under sustained CPU load, whereas AWS maintains moderate utilization when running the same NLP container.

\begin{table}[htbp]
    \centering
    \renewcommand{\arraystretch}{1.2}
    \caption{Comparison of system metrics during NLP processing.}
    \label{tab:system-metrics-comparison}
    \vspace{0.5em}
    \begin{tabular}{|l|l|l|}
        \hline
        \textbf{Metric} & \textbf{Local Machine} & \textbf{AWS Nitro Enclave} \\ \hline
        Docker CPU (\%)         & 1072.87 & 102.13  \\ \hline
        Docker Memory (MB)      & 1551.21 & 1146.89 \\ \hline
        Docker Processes        & 37      & 6       \\ \hline
        Host CPU Usage (\%)     & 90.98   & 51.81   \\ \hline
        Host Memory Used (MB)   & 1271    & 9529    \\ \hline
        Host Load Average       & 8.91    & 1.03    \\ \hline
    \end{tabular}
\end{table}

For execution time and scalability, Figs.~\ref{fig:execution-time-small} and~\ref{fig:execution-time-large} compare completion times across input sizes.

\begin{itemize}
    \item \textbf{Small/Moderate batches (5--40 files).}
    AWS completes inference in roughly 9--10\,s at the lower range (10-file batches), whereas the local machine ranges between 10--15\,s. As file counts increase, AWS times remain comparatively stable, reaching \(\sim 23\)\,s at 40 files, while the local host exceeds 25\,s.
    \item \textbf{Larger batches (50--100 files).}
    Beyond 50 files, gaps widen: the local setup reaches \(\sim 40\)\,s at 50 files and peaks above 70\,s at 90--100 files. AWS processes 50 files in \(<30\)\,s and scales to \(\sim 50\)\,s for 100 files, reflecting a more efficient response curve.
\end{itemize}

\noindent\textbf{Analysis.}
AWS Nitro Enclaves consistently outperform the local setup for NLP-based skill extraction under comparable CPU specifications. Virtualization and scheduling characteristics, together with I/O behavior, yield faster runtimes and lower load averages. The local machine’s sustained CPU saturation elevates the load average, limiting parallel-processing headroom. In contrast, AWS operates at moderate CPU usage and exhibits headroom to scale with minimal performance loss. Overall, pairing the containerized NLP pipeline with enclave execution is effective; local infrastructure remains suitable for small/development workloads, whereas the AWS Nitro Enclave provides a more robust foundation for high-volume tasks.

\begin{figure}[htbp]
    \centering
    \includegraphics[width=1\linewidth]{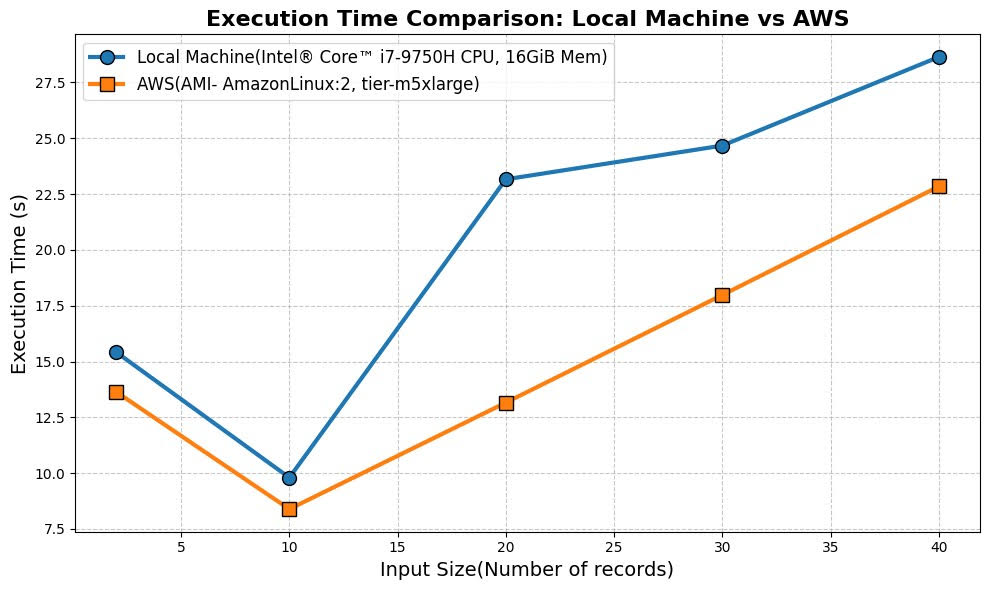}
    \caption{Execution time for small/moderate batches (5--40 files): local machine vs.\ AWS Nitro Enclave.}
    \label{fig:execution-time-small}
\end{figure}

\begin{figure}[htbp]
    \centering
    \includegraphics[width=1\linewidth]{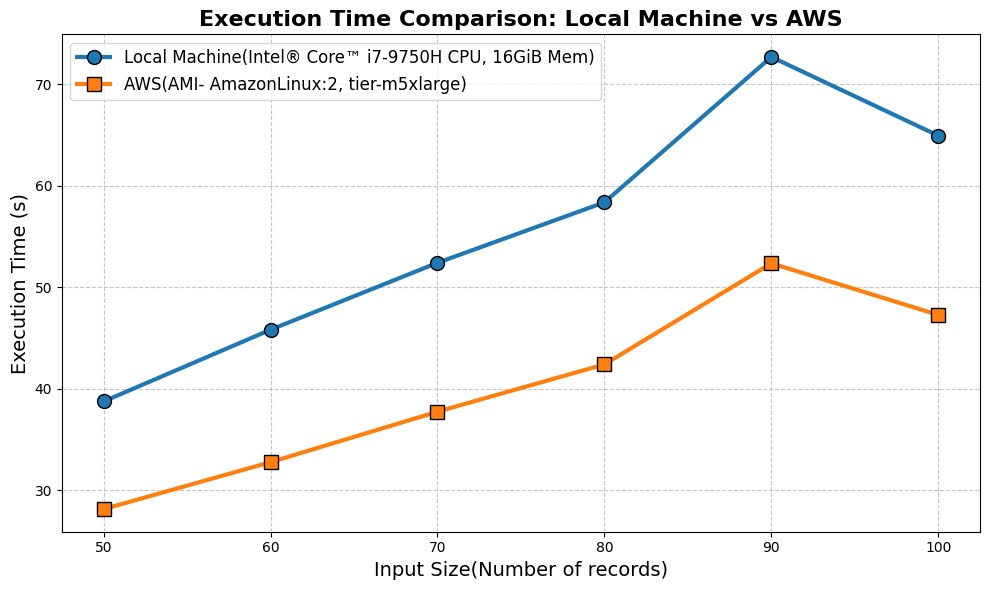}
    \caption{Execution time for larger batches (50--100 files): local machine vs.\ AWS Nitro Enclave.}
    \label{fig:execution-time-large}
\end{figure}

\noindent\textbf{Validation posture.}
The mapping follows the validated Syllabus--to--O*NET methodology; consequently, external validation evidence from prior work is inherited. In this prototype, a stability check across repeated runs observed \(<5\%\) variance among top-ranked skills, indicating consistency of the derived vectors under identical inputs and configuration.

\section{Discussion}
\label{discussion}
Zero-knowledge proofs (ZKPs) provide an alternative privacy mechanism to TEE for LER workflows. With ZKPs, verifiers can be convinced that a predicate holds (e.g., grade-point average exceeds a threshold) without access to underlying transcript data, thereby removing hardware trust assumptions and enabling strong cryptographic privacy guarantees. The trade-off is computational overhead and proving-time latency that may be nontrivial for interactive scenarios. TEEs, by contrast, offer low-latency execution and straightforward integration with existing credential flows but introduce dependence on hardware roots of trust and residual exposure to microarchitectural side channels. A pragmatic path is a hybrid: TEEs are used for high-throughput, real-time operations (ingestion, normalization, skills-only matching), while ZKPs are applied selectively to privacy-critical predicates or cross-domain verifications. Such a design preserves practicality while strengthening end-to-end privacy.

As the evidence set broadens to include job postings and informal artifacts, maintaining NLP accuracy and stability becomes central. Transformer-based methods exhibit strong semantic coverage yet require domain adaptation as terminology evolves. A human-in-the-loop feedback channel—where learners, educators, and employers validate extracted skills, flag misclassifications, and approve corrections—can be used to refine model prompts, filters, and thresholds over time. Lightweight interfaces for rating/tagging and audited correction logs support continuous improvement without relaxing privacy guarantees, because raw artifacts remain enclave-confined and only attested updates are released.

Operational controls reinforce compliance-oriented principles. Selective disclosure enforces data minimization and purpose limitation; consent is obtained at upload; and raw records and keys remain enclave-confined. Presentations are short-lived and must carry fresh revocation status within a window $\Delta$; derivative credentials embed provenance (inputs, code measurement, timestamp) so that disputes trigger re-issue and prior credentials are revoked but auditable. These controls align with common data-protection tenets (consent, minimization, freshness, retention), while jurisdiction-specific requirements (e.g., FERPA/GDPR scoping and retention periods) are left to deployment-time policy and legal review.

\noindent\textbf{Scope and Limitations.} Several limitations remain. External validity is constrained by the evaluation corpus: robustness to noisy or partial documents outside computer science was not benchmarked here and will require discipline-specific domain adaptation and validation. The guarantees described mitigate bias introduced at screening time by non-skill fields; upstream disparities in skill acquisition are not addressed by the matcher itself. Side-channel hardening beyond constant-time operations and minimal state exposure was not modeled formally; adoption of published mitigations and routine measurement/patch cycles is recommended for production deployments.

\noindent\textbf{External validity beyond computer science.}
Although the empirical evaluation used a computer-science corpus, the pipeline is taxonomy-agnostic and maps learning outcomes to O*NET skills with documented coverage across numerous fields of study. Portability to additional majors is therefore expected, provided (i) sufficient syllabus/learning-outcome availability and (ii) routine domain adaptation of filters and thresholds. In future work, discipline-specific validation (faculty spot-checks and sub-sampling stability analyses) will be repeated for non-CS fields to establish cross-domain reliability.

\section{Conclusion}
\label{conc}

In this paper, we introduced an AI-enabled, decentralized LER framework integrating trusted execution environments, decentralized identifiers, and NLP-based skill extraction has been presented. By incorporating both institutional records and informal learning artifacts, the framework enables holder-controlled selective disclosure and skills-only matching, thereby reducing opportunities for screening bias while preserving confidentiality. Experimental evidence on representative workloads indicates that offloading NLP processing to a cloud enclave improves performance and scalability, supporting the practicality of enclave-backed deployments.

Future extensions include broader domain coverage and robustness assessments beyond computer science, human-in-the-loop feedback to refine extraction and thresholds, and parsing of job postings for tighter alignment with market demand. A hybrid design that combines enclaves for real-time operations with zero-knowledge proofs for selective, privacy-critical predicates is expected to strengthen privacy without sacrificing usability. Taken together, secure computation, standardized skill mapping, and self-sovereign control position the proposed LER framework as a flexible, privacy-preserving foundation for modern education and workforce alignment.

\section*{Acknowledgment}
This work was supported in part by the Walmart Foundation. 

\bibliographystyle{IEEEtran}
\bibliography{ref}

\end{document}